\documentclass[journal,twoside]{IEEEtran}
\usepackage{amsmath}
\usepackage{mathrsfs}
\usepackage{amsfonts}
\usepackage{graphicx}
\usepackage{citesort}
\usepackage{amssymb}
\usepackage{amsfonts}
\usepackage{amsmath}
\usepackage{epsfig}
\usepackage{color}
\usepackage{fancybox}
\usepackage{textcomp}
\usepackage{multirow}
\usepackage{setspace}
\usepackage{psfrag}
\usepackage{bbding}

\usepackage{chngpage}
\usepackage{array}

\newtheorem{Theorem}{Theorem}

\newtheorem{Remark}{Remark}

\newtheorem{Corollary}{Corollary}

\newtheorem{Proposition}{Proposition}

\makeatletter
\def\hlinewd#1{%
  \noalign{\ifnum0=`}\fi\hrule \@height #1 \futurelet
   \reserved@a\@xhline}
\makeatother

\def\Complex{{\rm\rule[.23ex]{.03em}{1.1ex}\kern-.3em{C}}}

\newcommand{\be}{\begin{equation}}
\newcommand{\ee}{\end{equation}}
\newcommand{\bea}{\begin{eqnarray}}
\newcommand{\eea}{\end{eqnarray}}
\newcommand{\benum}{\begin{enumerate}}
\newcommand{\eenum}{\end{enumerate}}

\newcommand{\qa}{{\bf a}}

\newcommand{\qg}{{\bf g}}
\newcommand{\qh}{{\bf h}}

\newcommand{\qm}{{\bf m}}

\newcommand{\qw}{{\bf w}}
\newcommand{\qx}{{\bf x}}

\newcommand{\qz}{{\bf z}}
\newcommand{\qA}{{\bf A}}

\newcommand{\qI}{{\bf I}}

\newcommand{\qP}{{\bf P}}
\newcommand{\qQ}{{\bf Q}}
\newcommand{\qR}{{\bf R}}

\newcommand{\qW}{{\bf W}}

\newcommand{\qPhi}{{\boldsymbol \Phi}}

\newcommand{\bbC}{{\mathbb C}}

\newcommand{\calL}{{\mathcal L}}

\newcommand{\tr}{{\sf tr}}
\newcommand{\Ex}{{\sf E}}

\newcommand{\varx}{{\sf var}}

\newcommand{\argmin}{\operatornamewithlimits{arg\, min}}

\newcommand{\tu}{\tau_{u}}

\begin{document}

\title{Energy Efficient Downlink Transmission for Multi-cell Massive DAS with Pilot Contamination}

\author{Jun Zuo,
        Jun Zhang,
        ~\IEEEmembership{Member,~IEEE,}
        Chau Yuen,
        ~\IEEEmembership{Senior Member,~IEEE,}\\
        Wei Jiang,
        ~\IEEEmembership{Member,~IEEE,}
        and Wu Luo
        ~\IEEEmembership{Member,~IEEE}
\thanks{Manuscript received November 18, 2015; revised February 29, 2016 and April 4, 2016; accepted April 7, 2016. The work of J. Zuo, W. Jiang, and W. Luo is supported by the National Natural Science Foundation of China under Grant 61171080. The work of C. Yuen is supported in part by Singapore A*STAR SERC Project under Grant 142 02 00043 and the National Natural Science Foundation of China under Grant 61550110244. The work of J. Zhang is supported in part by the Natural Science Foundation through the Jiangsu Higher Education Institutions of China under Grant 15KJB510025, the Natural Science Foundation Program through Jiangsu Province of China under Grant BK20150852, and Jiangsu Planned Projects for Postdoctoral Research Funds under Grant 1501018A. The review of this paper was coordinated by Dr. Tomohiko Taniguchi.}
\thanks{Copyright (c) 2016 IEEE. Personal use of this material is permitted. However, permission to use this material for any other purposes must be obtained from the IEEE by sending a request to pubs-permissions@ieee.org.}
\thanks{J. Zuo, W. Jiang, and W. Luo are with the State Key Laboratory of Advanced Optical Communication Systems and Networks, Peking University, 100871, China, E-mail: \{zuojun, jiangwei, luow\}@pku.edu.cn.}
\thanks{J. Zhang is with Jiangsu Key Laboratory of Wireless Communications, Nanjing University of Posts and Telecommunications, Nanjing 210003, China, E-mail: zhangjun@njupt.edu.cn.}
\thanks{C. Yuen is with Singapore University of Technology and Design, Singapore, E-mail: yuenchau@sutd.edu.sg.}
\thanks{*The corresponding author is W. Jiang.}
}
\markboth{IEEE TRANSACTIONS ON VEHICULAR TECHNOLOGY ,~Vol.~$\times$, No.~$\times$, $\times\times$~2016} {Zuo \MakeLowercase{\textit{et al.}}: Multi-Cell Multi-User Massive MIMO Transmission with Downlink Training and Pilot Contamination Precoding}

\maketitle

\begin{abstract}
In this paper, we study the energy efficiency (EE) of a downlink multi-cell massive distributed antenna system (DAS) in the presence of pilot contamination (PC), where the antennas are clustered on the remote radio heads (RRHs). We employ a practical power consumption model by considering the transmit power, the circuit power, and the backhaul power, in contrast to most of the existing works which focus on co-located antenna systems (CAS) where the backhaul power is negligible. For a given average user rate, we consider the problem of maximizing the EE with respect to the number of each RRH antennas $n$, the number of RRHs $M$, the number of users $K$, and study the impact of system parameters on the optimal $n$, $M$ and $K$. Specifically, by applying random matrix theory, we derive the closed-form expressions of the optimal $n$, and find the solution of the optimal $M$ and $K$, under a simplified channel model with maximum ratio transmission. From the results, we find that to achieve the optimal EE, a large number of antennas is needed for a given user rate and PC.
As the number of users increases, EE can be improved further by having more RRHs and antennas.
Moreover, if the backhauling power is not large, massive DAS can be more energy efficient than massive CAS. These insights provide a useful guide to practical deployment of massive DAS.
\end{abstract}


\section{Introduction}
With the rapid deployment of wireless communication systems, energy efficiency (EE) becomes a key concern from the viewpoint of green communication \cite{Tombaz-11WCM,Andrews-14JSAC}. Recently, massive multiple-input multiple-output (MIMO) systems, where a large number of antennas are deployed at the base station (BS), have attracted a great deal of research interest \cite{Marzetta-10TWC,Rusek-13SPM,Suraweera-13ICC,Hoydis-13JSAC,ZhangJun-13JSAC,Larsson-14ComM,Lu-14JSTSP,Sanguinetti-15JSAC,Sadeghi-15GC,ZhangJun-15TWC}. Massive MIMO is acknowledged as a promising technology to improve both the spectral efficiency (SE) and EE with the advantages of asymptotically negligible fast fading, noise free channels, and arbitrarily small transmit power \cite{Yang-13JSAC,Gao-15CL,Chen-15TWC}. The major bottleneck of improving the SE in massive MIMO is the so-called pilot contamination (PC) effect, which is caused by using the non-orthogonal uplink pilot sequences at different users \cite{Marzetta-10TWC,Fernandes-13JSAC}. On the other hand, distributed antenna systems (DAS), where antennas of the interested cell can either be fully distributed within the cell \cite{Liu-14TSP,Joung-14JSTSP,LinYicheng-14TWC} or clustered at remote radio heads (RRHs) \cite{Truong-13AC,Onireti-13TCOM,Wangdongming-13JSAC,Sun-15IET}, is proven to be efficient to improve the EE and coverage by shortening the average distance between the transmitters and users, and thus lowering the transmit power \cite{Heath-11TSP,Dailin-14TWC}. It is expected that combining DAS with massive MIMO by scaling up the number of antennas in DAS, i.e., massive DAS, can further enhance the system performance \cite{Truong-13AC,Liu-14TSP}.

The EE analysis and optimization problems in massive MIMO systems have been recently considered in \cite{Ngo-13TCOM,Liu-15arXiv,EmilB-15TWC,Yang-15VTC,Mohammed-14TC,Joung-14JSTSP,Hechunglong-12VTC}. For the massive co-located antenna systems (CAS), the power scaling law and trade-off between EE and SE for uplink transmission were analyzed in \cite{Ngo-13TCOM}, where only the transmit power was considered when evaluating the EE. In \cite{Liu-15arXiv}, the authors investigated the EE of downlink multi-cell massive CAS by optimizing the transmit power for given numbers of BS antennas and users. Focusing on zero forcing (ZF) processing in single-cell systems with perfect channel state information (CSI) at the BS, an EE optimization problem was discussed in \cite{EmilB-15TWC} to find the optimal numbers of BS antennas, users, and transmit power. The authors of \cite{Yang-15VTC} optimized the number of BS antennas to maximize EE when PC was negligible, and provided the explicit formulas of the optimal number of BS antennas in single-cell case. The impact of transceiver power consumption on the EE of the ZF detector in the uplink single-cell massive CAS was discussed in \cite{Mohammed-14TC}.

For the DAS, in \cite{Joung-14JSTSP}, the design of precoding matrix, antenna selection matrix, and power control matrix to optimize the EE in single-cell downlink massive DAS was studied.
A comparative EE study of uplink transmission between DAS and CAS was considered in \cite{Hechunglong-12VTC} under a power consumption model considering transmit power and circuit power, and revealed that DAS can improve the EE when compared to CAS.

However, most of these works only focused on the single-cell scenario for analytical tractability. To the best of the authors' knowledge, there is limited study analyzing the EE in multi-cell massive DAS and taking into account the impact of PC.
To this end, we take into account PC and investigate the EE in the downlink \emph{multi-cell} massive DAS, where the antennas are clustered at RRHs.
Moreover, the power consumption model is important when evaluating the EE. In this paper, we adopt a power model where the transmit power, the circuit power, and the backhaul power are considered \cite{Joung-14JSTSP,Onireti-13TCOM,Hechunglong-12VTC,Ngo-13TCOM}. The comparison among our work and previous work are listed in Table~\ref{Related Work}, where ``UL" and ``DL" denote uplink and downlink, respectively.
\begin{table*}[!t]
\centering
\caption{Comparison of Related Work of EE in Massive MIMO} \label{Related Work}
\begin{adjustwidth}{0cm}{-0.5cm}
\begin{tabular}{|m{.08\linewidth}<{\centering}| m{.11\linewidth}<{\centering}| m{.07\linewidth}<{\centering}| m{.07\linewidth}<{\centering}| m{.04\linewidth}<{\centering}| m{.44\linewidth}<{\centering}|}
 \hline Work & CAS/DAS & Cell & UL/DL & PC & Main Contribution\\
\hline \cite{Ngo-13TCOM} & CAS &
 Single\& Multi
 & UL &  ${\surd}$ & Study the power scaling law and trade-off between EE and SE\\
\hline \cite{Liu-15arXiv} & CAS & Multi\ & DL & ${\surd}$ & Optimize the transmit power \\
\hline \cite{EmilB-15TWC} & CAS & Single & UL\&DL &${\times}$ & Optimize the numbers of BS antennas, users, and the transmit power \\
\hline \cite{Yang-15VTC} & CAS & Multi & DL & ${\times}$ & Optimize the number of BS antennas\\
\hline \cite{Mohammed-14TC}& CAS & Single & UL & ${\times}$ & Study the impact of transceiver power consumption on the EE \\
\hline \cite{Joung-14JSTSP} & DAS & Single & DL & ${\times}$ & Design the precoding matrix, antenna selection matrix, and power control matrix\\
\hline \cite{Hechunglong-12VTC}& CAS\&DAS & Single &UL &${\times}$ & Compare the EE between DAS and CAS\\
\hline Proposed& DAS & Multi & DL & ${\surd}$ & Optimize the antenna number of each RRH, the numbers of RRHs and users \\
\hline
 \end{tabular}
\end{adjustwidth}
\end{table*}

In particular, we are interested in the following problems. For a given average uniform rate, to achieve optimal EE, how many antennas should be employed by each RRH? How many RRHs should be deployed? What is the optimal number of users? And how the optimal numbers are affected by different parameters, including the channel correlation, the channel gain, the power consumption parameters, and the PC? Per-user power optimization is an important issue in EE maximization problem. Here, this issue is not involved so as to study the effects of the number of antennas, RRHs, and users on EE in a standalone manner and draw basic insights. The discussions on EE optimization of per-user power can be found in \cite{Zhao-13VTC,Chien-16arXiv,Li-14WCNC}.
The EE optimization problem in general are difficult problems when taking into account the imperfect CSI at the RRHs and the effect of multi-cell PC, which makes it difficult to analyze. To solve the problems, we first use random matrix theory to reduce random channel gains to deterministic statistical information \cite{Wagner-12IT,Hoydis-13JSAC,ZhangJun-13TWC}. Second, we consider a simplified channel model to facilitate the analysis. By doing so, a closed-form expression on the optimal antenna number of each RRH is derived, the form of solution for the optimal number of users is given, and finally the optimal number of RRHs is obtained through one-dimensional search.
From the results, we find that to achieve the optimal EE, a large number of antennas is needed for a given user rate and PC.
As the number of users increases, EE can be improved further by having more RRHs and antennas.
Moreover, if the backhauling power is not large, massive DAS can be more energy efficient than massive CAS. These insights provide a useful guide to practical deployment of massive DAS.

The rest of the paper is organized as follows. The system model and power consumption model are described in Section II. In Section III, the asymptotic EE is derived, and this is then used in Section IV to obtain the optimal antenna number of each RRH, the optimal number of RRHs, and the optimal number of users that maximize the EE. We then analyze how these optimal numbers are affected by other system parameters. Simulation results are presented in Section V to validate the analysis, followed by conclusions in Section VI.

\emph{Notation:} Boldface uppercase and lowercase letters denote matrices and vectors, respectively. An $N \times N$ identity matrix is denoted by ${\bf{I}}_N$, while an all-zero matrix is denoted by ${\bf 0}$, and an all-one matrix by ${\bf 1}$. The superscripts $(\cdot)^{H}$, $(\cdot)^{T}$, and $(\cdot)^{*}$ stand for the conjugate-transpose, transpose, and conjugate operations, respectively. $\Ex\{\cdot\}$ means the expectation operator, and $\varx \{\cdot \}$ denotes the variance. We use $\tr\{\qA\}$ to denote the trace of matrix $\qA$ and diag\{$\qa$\} to denote a diagonal matrix with vector $\qa$ along its main diagonal. The notation $| \cdot|$ and $\|\cdot\|$ denote the absolute value of a variable and the two-norm of a matrix, respectively. $\qx\sim \mathcal{CN}\left( \qm, \qQ \right)$ defines a vector of jointly circularly symmetric complex Gaussian random variables with mean value $\qm$ and covariance matrix $\qQ$.

\section{System Model and Power Consumption Model}
\subsection{System Model}
Consider the downlink of a cellular network with $L$ non-coordinated cells, where each cell consists of $M$ RRHs and $K$ randomly distributed single-antenna users. The RRHs and users in cell $l$ are labeled as ${\sf RRH}_{l,1}, \dots, {\sf RRH}_{l,M}$ and ${\sf UE}_{l,1}, \dots, {\sf UE}_{l,K}$, respectively. $N$ $(N \gg K)$ antennas in a cell are evenly divided among RRHs, such that each RRH equips $n=N/M$ antennas. The $M$ RRHs in the same cell are connected to a baseband processing unit (BPU), where the main operations, including data processing and management processing are implemented. The system works in time-division duplexing (TDD) mode so that the channels between uplink and downlink are reciprocity. An example of 7-RRH massive DAS is shown in Fig.~\ref{systemmodel}, in each cell, there is one RRH in the cell center and six RRHs uniformly spaced on a circle of distance 2/3  radius away from the cell center.
\begin{figure}
\centering
\includegraphics[width=0.5\textwidth]{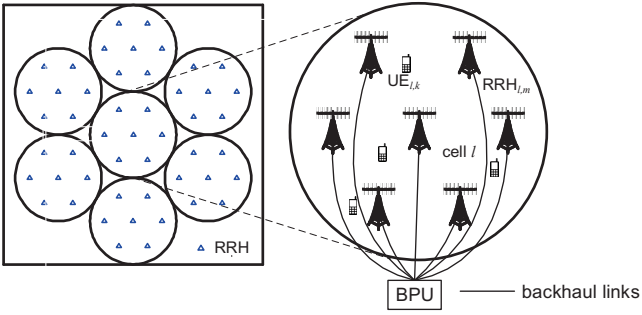}
\caption{{System model of multi-cell massive DAS.}}\label{systemmodel}
\end{figure}

The channel between ${\sf RRH}_{l,m}$ and ${\sf UE}_{j,k}$ is expressed as
\begin{equation}\label{eq:channel model}
\qg_{lmjk}=\qR_{lmjk}^{1/2}\qh_{lmjk},
\end{equation}
where $\qh_{lmjk}  \in \mathbb C^n$ is the small-scale fading channel vector, whose elements are independent and identically
distributed (i.i.d.) complex random variables with zero-mean and unit variance, and $\qR_{lmjk}= \Ex\{ \qg_{lmjk}\qg_{lmjk}^H\} \in \mathbb C^{n \times n}$ describes the spatial correlation and large-scale fading of the channel, which is a deterministic nonnegative definite matrix. $\qg_{ljk}=[\qg_{l1jk}^T, \qg_{l2jk}^T, \dots, \qg_{lMjk}^T]^T \in \mathbb C^{N}$ is the channel vector between all the $M$ RRHs in cell $l$ and ${\sf UE}_{j,k}$.

During uplink pilot transmission phase, all users simultaneously transmit pilot sequences with length $\tau_u=\psi K$ and power $p_u$, where $\psi \ (\psi \geq 1)$ is the pilot reuse factor. We assume that the pilot sequences of users in the same cell are pairwisely orthogonal, and the pilot reuse in different cells are indicated by $\psi$. For instance, $\psi=L$ allows assigning all cells orthogonal pilot sequences, where the PC is absent, and $\psi=1$ means the worst case scenario of PC, where every cell reuses the same set of pilot sequences. If $\calL_j$ is the set of cells sharing the same set of pilot sequence as cell $j$, then the number of users sharing the same pilot sequence as ${\sf UE}_{j,k}$ is $L/\psi$. Given the statistical knowledge of the channel, i.e., $\qR_{jmjk}$ and $\qQ_{jmjk}$, the MMSE estimate of $\qg_{jmjk}$ at the BPU in cell $j$ can be expressed as \cite{Kay-93,Hoydis-13JSAC,Marzetta-06ACSSC}
\begin{align}\label{eq:channel estimate}
&\hat{\qg}_{jmjk} \nonumber\\
=&\qR_{jmjk}\qQ_{jmjk}\bigg(\qg_{jmjk}+\sum_{l \in \calL_j \backslash \{j\}}\qg_{jmlk}+\frac{1}{\sqrt{p_u \tau_u}}\qz_{jmk}\bigg),
\end{align}
where $\qz_{jmk} \sim \mathcal{CN} (0, \sigma^2 \qI_n)$ denotes the Gaussian noise, and $\qQ_{jmjk}= \big( \frac{\sigma^2}{p_u \tu} \qI_{n} + \sum\limits_{l \in \calL_j} \qR_{jmlk}\big)^{-1}$.
From \eqref{eq:channel estimate}, it can be verified that $\hat \qg_{jmjk} \sim \mathcal{CN} ({\bf 0}, \qPhi_{jmjk})$ with $\qPhi_{jmjk}=\qR_{jmjk} \qQ_{jmjk} \qR_{jmjk}$ \cite{Hoydis-13JSAC}. The second term of the right-hand side of \eqref{eq:channel estimate} represents the PC from other cells.

For downlink data transmission, we assume that all the $M$ RRHs in each cell jointly serve the $K$ users within the cell. The downlink signal received by ${\sf UE}_{j,k}$ is given by
\begin{equation}\label{eq:received signal}
  y_{jk}=\sqrt{p_d}\sum_{l=1}^L \sum_{m=1}^M \qg_{lmjk}^T \qx_{lm} + z_{jk},
\end{equation}
where $p_d$ is the transmit power, $z_{jk} \sim \mathcal{CN} (0, \sigma^2)$ is the noise, and $\qx_{lm} \in \mathbb C^n$ is the transmit signal of ${\sf RRH}_{l,m}$, which can be expressed as
\begin{equation}
  \qx_{lm}=\sqrt{\lambda_{l}}\sum_{i=1}^K \qw_{lmi} s_{lmi},
\end{equation}
where $\qw_{lmi} \in \bbC^{n}$ is the precoding vector for ${\sf UE}_{l,i}$, $\lambda_{l}$ normalizes the transmit power in cell $l$ so that $\Ex \big\{ \frac{p_d}{K} \sum_{m=1}^M \qx_{lm}^H \qx_{lm}\big\}=p_d$, and $s_{lmi}$ is the information-bearing signal with $\Ex \left\{s_{lmi} s_{lmi}^*\right\}=1$.

We adopt the same assumption as in \cite{Hoydis-13JSAC,Jose-11TWC} that the channel estimates are available at the BSs or the BPUs, and only the statistical properties of the channel $\Ex \{\qg_{jmjk}^T \qw_{jmk}\}$, $m=1,2,\dots,M$, are known at the UEs for detecting its desired signal. Therefore, the received signal in \eqref{eq:received signal} can be rewritten as
\begin{align}\label{eq: rewritten the received signal}
y_{jk}=&\sqrt{p_d \lambda_j}\sum_{m=1}^M \Ex \left\{\qg_{jmjk}^T \qw_{jmk} \right\}s_{jmk} \nonumber\\
       &+\sqrt{p_d \lambda_j} \sum_{m=1}^M\big( \qg_{jmjk}^T \qw_{jmk}
       -\Ex \left\{ \qg_{jmjk}^T \qw_{jmk} \right\}\big) s_{jmk}\nonumber\\
       &+ \sqrt{p_d \lambda_j} \sum_{i \neq k}\sum_{m=1}^M \qg_{jmjk}^T \qw_{jmi} s_{jmi} \nonumber\\
       &+ \sum_{l \neq j} \sqrt{p_d \lambda_l}\sum_{i=1}^K \sum_{m=1}^M \qg_{lmjk}^T \qw_{lmi} s_{lmi} + z_{jk}.
\end{align}
In \eqref{eq: rewritten the received signal}, the first term is the desired signal, and other terms can be treated as the effective noise. The signal-to-interference-plus-noise ratio (SINR) can be given by
\begin{equation}\label{eq:SINR}
  \mbox{SINR}_{jk}=\frac{\lambda_j \left|\sum\limits_{m=1}^M \Ex \left\{\qg_{jmjk}^T \qw_{jmk} \right\}\right|^2}{\lambda_j \varx \left\{\sum\limits_{m=1}^M \qg_{jmjk}^T \qw_{jmk} \right\}+\mbox{SCI}_{jk} + \mbox{ICI}_{jk} + \frac{\sigma^2}{p_d}},
\end{equation}
where the interference from users in the same cell (SCI) and the inter-cell interference (ICI) are, respectively, given by
\begin{subequations}
\begin{align}
\mbox{SCI}_{jk}&=\lambda_j \sum_{i \neq k} \Ex\left\{\left|\sum_{m=1}^M \qg_{jmjk}^T \qw_{jmi}\right|^2\right\}, \\
\mbox{ICI}_{jk}&=\sum_{l \neq j} \sum_{i=1}^K \lambda_l \Ex\left\{\left|\sum_{m=1}^M \qg_{lmjk}^T \qw_{lmi}\right|^2\right\}.
\end{align}
\end{subequations}
As shown in \cite{Hoydis-13JSAC,Jose-11TWC}, the downlink SE of cell $j$ can be expressed as
\begin{equation}\label{achievable rate}
R_j=\frac{T-\tu}{T}\sum_{k=1}^K \log_2 \left( 1+ \mbox{SINR}_{jk} \right)\ (\text {in bits/s/Hz}),
\end{equation}
where $T$ is the channel coherence interval in symbols.

\subsection{Practical Power Consumption Model}
It is necessary to use a practical power consumption model for evaluating the EE accurately. Based on \cite{EmilB-15TWC,Onireti-13TCOM}, the total power consumed for the downlink transmission of a given cell can be modeled as the sum of a fixed power part, the circuit power, the transmit power, and backhaul inducing power:
\begin{equation}\label{eq:P_total}
  P_{\text{Total}}=P_{\text{FIX}}+N P_{\text{RRH}}+\frac{T-\tau_u}{T} \frac{p_d}{\zeta} K+P_{\text{BH}},
\end{equation}
where $P_{\text{FIX}}$ accounts for the static circuit power consumption, $P_{\text{RRH}}$ is the power required to run the internal RF components of each RRH antenna, $p_d$ is the average transmit power normalized to users, $\zeta$ is the amplifier efficiency, and $P_{\text{BH}}$ is the power consumed by backhaul links.

The backhaul inducing power in DAS might be significant since all RRHs are connected to their BPUs through high-speed backhaul links such as optical fiber. However, in CAS, the power consumption of backhaul is much less because the data processing can be done in the BS that is close to the antennas.
In massive DAS, the power consumption of backhaul for connecting $M$ RRHs to BPU is modeled as \cite{Onireti-13TCOM,EmilB-15TWC}
\begin{equation}
   P_{\text{BH}}^{\text{DAS}}= M (P_0 +  R B P_{\text{BT}}),
\end{equation}
where $P_0$ is a fixed power consumption of each backhaul, $R$ is the spectral efficiency (in bits/s/Hz), $B$ is the system bandwidth, and $P_{\text{BT}}$ is the traffic dependent power (in Watt per bit/second).

Given the system model and the power consumption model, we will adopt maximum-ratio transmission (MRT) as an example to analyze the EE in the following section. Our analysis and design are
also applicable when other beamforming strategies are adopted by RRHs.

\section{Asymptotic Energy Efficiency}
In this section, we first derive the deterministic expressions of the asymptotic SE and EE. The derivations are based on the assumption that the number of RRHs $M$ is finite, while the antenna number of each RRH $n$ and the number of users $K$ approach to infinity at a fixed ratio $n/K$. Since the derived deterministic expressions are accurate even in non-asymptotic regime, we can use them for EE optimization in practical case, which will be shown in Section IV.

If MRT beamforming is adopted in transmission, the precoding vector is given by
\begin{equation}
\qw_{lmi}={\hat{\qg}_{lmli}}^*.
\end{equation}

In \cite[Theorem 4]{Hoydis-13JSAC}, the deterministic approximations of SINR with MRT beamforming of co-located multi-cell massive MIMO system has been derived. However, the distributed massive MIMO system under considered is a more general scenario.
To derive the deterministic equivalent of SINR, we make the following assumptions:
\begin{itemize}
\item The spectral norm of $\qR_{lmjk}$, $\forall l,m,j,k$, is uniformly bounded with respect to $n$.
\item The trace of $\qR_{lmjk}$, $\forall l,m,j,k$, scales linearly with $n$.
\item The channel estimate $\hat{\qg}_{jmjk}$, the estimate error $\tilde{\qg}_{jmjk}$, and the noise $\qz_{jmk}$, $\forall j,m,k$, are mutually independent.
\end{itemize}

\begin{Proposition}\label{Proposition:SINR}
As $n,K \rightarrow \infty$, user's SINR is approximated by a deterministic equivalent such that
\begin{equation}
 \mbox{SINR}_{jk}-\overline {\mbox{SINR}}_{jk} \xrightarrow{a.s.} 0,
\end{equation}
where $\overline {\mbox{SINR}}_{jk}$ is given in by \eqref{eq:SINR_f_RMT}, shown at the top of next page,
\begin{figure*}[!t]
\begin{equation}\label{eq:SINR_f_RMT}
\overline {\mbox{SINR}}_{jk}=\frac{\bar \lambda_j \left( \frac{1}{n} \sum\limits_{m=1}^M \tr \qPhi_{jmjk}\right)^2}{\sum\limits_{l \in \calL_j \backslash \{j\}} \bar \lambda_l \left| \frac{1}{n} \sum\limits_{m=1}^M \tr \qPhi_{lmjk}\right|^2 + \frac{1}{n}\sum\limits_{l=1}^L \sum\limits_{m=1}^M \sum\limits_{i=1}^K \bar \lambda_{l}\frac{1}{n} \tr \qR_{lmjk} \qPhi_{lmli} +\frac{\sigma^2}{p_d n}}.
\end{equation}
\hrulefill
\end{figure*}
with $\bar \lambda_l=\big(\frac{1}{K} \sum_{i=1}^{K} \frac{1}{n} \sum_{m=1}^M \tr \qPhi_{lmli}\big)^{-1}$, and the notation ``$\xrightarrow{a.s.}$'' denotes the almost sure (a.s.) convergence.
\end{Proposition}
\begin{sketch}
Dividing the denominator and numerator of $\mbox{SINR}_{jk}$ by $\frac{1}{n}$, we obtain the asymptotic results of each item in $\mbox{SINR}_{jk}$ as follows:
$\lambda_j \big|\sum\limits_{m=1}^M \Ex \big\{\qg_{jmjk}^T \qw_{jmk} \big\}\big|^2 \xrightarrow[n \rightarrow \infty]{a.s.} \lambda_j \big( \frac{1}{n} \sum\limits_{m=1}^M \tr \qPhi_{jmjk}\big)^2$,
$\mbox{SCI}_{jk} + \mbox{ICI}_{jk} \xrightarrow[n \rightarrow \infty]{a.s.} \sum\limits_{l \in \calL_j \backslash \{j\}} \bar \lambda_l \left| \frac{1}{n} \sum\limits_{m=1}^M \tr \qPhi_{lmjk}\right|^2 + \frac{1}{n}\sum\limits_{l=1}^L \sum\limits_{m=1}^M \sum\limits_{i=1}^K \bar \lambda_{l}\frac{1}{n} \tr \qR_{lmjk} \qPhi_{lmli}$,
and \\ $\frac{1}{N} p_{j,k} \lambda_j \varx \big\{\sum\limits_{m=1}^M \qh_{jmjk}^T \qw_{jmk} \big\} \xrightarrow[n \rightarrow \infty]{a.s.} 0$.
For the detailed proof of this proposition, please refer to the proof of \cite[Theorem 4]{Hoydis-13JSAC}.
\end{sketch}

The downlink EE of cell $j$ is defined as the downlink SE divided by the total power consumed in downlink transmission of cell $j$:
\begin{equation}\label{achievable energy efficiency}
 \eta_j\triangleq \frac{B R_j} {P_{\text{Total}}(R_j )} \ (\text {in bits/Joule}).
\end{equation}

Proposition~\ref{Proposition:SINR} indicates that user's SINR can be approximated by its deterministic equivalent without the needs of knowing the instantaneous channel. Based on continuous mapping theorem, we have the following almost sure convergence \cite{ZhangJun-13TWC}
\begin{equation}\label{eq:mapping theorem}
\eta_j-\overline \eta_j \xrightarrow{a.s.} 0,
\end{equation}
where $ \overline \eta_j=\frac{B \overline R_j}{P_{\text{Total}}(\overline R_j)}$, and $ \overline R_j=\frac{T-\tau_u}{T}\sum\limits_{k=1}^K \log (1+ \overline {\mbox{SINR}}_{jk})$.

In practice, the large-scale fading factors or the attenuation factors between different users and RRHs are not the same, however, this makes it very difficult (if not impossible) to investigate the EE and obtain basic insights. To tackle this issue, we consider a simplified channel model used in \cite{Hoydis-13JSAC,Ngo-13TCOM,Liu-15arXiv,Mohammed-14TC}, which is given by
\begin{equation}\label{eq:simplified model}
\qg_{lmjk}=\sqrt{\beta_{lmjk} \frac{n}{P}} \qA \tilde \qh_{lmjk}.
\end{equation}
The channel model in \eqref{eq:simplified model} is a particular physical channel model of \eqref{eq:channel model}. For large antenna systems, due to either insufficient antenna spacing or a lack of scattering, the channel correlation matrix $\qR_{lmjk}$ may not have full rank \cite{Ngo-11ICASSP}. The model in \eqref{eq:simplified model} is obtained by letting $\qR_{lmjk}^{1/2}=\sqrt{\beta_{lmjk} \frac{n}{P}} [\qA \ {\bf 0}_{n \times (n-P)}]$, where $\beta_{lmjk}$ is the large-scale fading factor, $\qA \in \mathbb C^{n \times P}$ is the array steering matrix\cite{Ngo-11ICASSP}, which describes the channel correlation and $P=\frac{n}{d}(d \geq 1)$ angles of arrival. As in \cite{Hoydis-13JSAC,Liu-15arXiv},
here $\qA$ is composed of $P$ columns of an arbitrary unitary $n \times n$ matrix, and $\qA$ can be given by different forms according to different physical channel models.
$\tilde \qh_{lmjk} \in \mathbb C^{P}$ is the small-scale fading channel vector, whose elements follow i.i.d. standard complex Gaussian distribution.
The large-scale fading factor is modeled as $\beta_{lmjk}=1/d_{lmjk}^\iota$, where $d_{lmjk}$ is the distance between ${\sf UE}_{j,k}$ and ${\sf RRH}_{l,m}$, and $\iota$ is the path-loss exponent.\footnote{The simplified model can be used because of the following two reasons. First, the number of degrees of freedom $P$, which depends on the scattering in the channel can be assumed as constant or to scale with the number of antennas $n$ \cite{Hoydis-13JSAC}. Second, the assumption that all users have the same correlation matrix reflects a worst-case performance because users instantaneous channel vectors are less orthogonal due to the same correlation matrix, which leads to large multi-user interference.}

Denote the index of the RRH in cell $j$ with minimum distance to ${\sf{UE}}_{j,k}$ as $\bar m_{jk}$. The average large-scale fading factor between ${\sf UE}_{j,k}$ and ${\sf RRH}_{j,\bar m_{jk}}$ (the average is taken over different users and different user locations) is related to both the number of RRHs $M$ and the radius of the cell. If $M$ is increased, or if the cell radius is decreased, the average distance between ${\sf UE}_{j,k}$ and ${\sf RRH}_{j,\bar m_{jk}}$ will be reduced. Assume that each cell is a circle with radius $R_c$, and the coverage area of each RRH is a circle with radius $r$. Then, $r$ can be approximated as ${R_c}/{\sqrt{M}}$. Since the average distance between ${\sf UE}_{j,k}$ and ${\sf RRH}_{j,\bar m_{jk}}$ is scaled with $r$, base on $\beta_{lmjk}=1/d_{lmjk}^\iota$, $\beta_{j\bar m_{jk}jk}$ is scaled with $M^\frac{\iota}{2}$. The average distances between ${\sf UE}_{j,k}$ and other $M-1$ RRHs in its cell (i.e., ${\sf RRH}_{jm}$, $m \neq \bar m_{jk}$), and the average distances between ${\sf UE}_{j,k}$ and RRHs in other cells (i.e., ${\sf RRH}_{lm}$, $l \neq j$), can be roughly treated as independent of $M$ and only determined by the cell radius $R_c$.

Based on the above analysis, $\beta_{lmjk}$ can be given by

\begin{align}\label{eq: interference factor}
 \beta_{lmjk}=\begin{cases}
M ^ \frac{\iota}{2} \beta, &\text{if} \quad j=l \ \text{and} \ m=\bar m_{jk}, \\
\alpha_1 \beta, &\text{if} \quad j=l \ \text{and} \ m \neq \bar m_{jk}, \\
\alpha_2 \beta, &\text{if} \quad j \neq l.
\end{cases}
\end{align}
where $\beta$ is the average large-scale fading with respect to different user locations, and it is determined by the cell radius and path-loss exponent. $\alpha_1$ ($0 \leq \alpha_1 \leq 1)$ represents the difference of large-scale fading factors from the nearest RRH and other $M-1$ RRHs in the cell, and $\alpha_2\ (0 \leq \alpha_2 \leq 1)$ can be named as inter-cell interference factor, which represents the difference of large-scaling factors from the nearest RRH and RRHs in other cells. When $M=1$ and $\alpha_1=0$, this model is consistent with the simplified model of CAS in \cite{Hoydis-13JSAC,Ngo-13TCOM,Liu-15arXiv}. With the simplified model, we have the following corollary.

\begin{Corollary}\label{Corollary: SINR of simplified model}
\textit{With the simplified model in \eqref{eq:simplified model}, the deterministic equivalent of user's SINR in \eqref{eq:SINR_f_RMT} can be written as}
\begin{equation}\label{eq:DE of simplified}
\overline {\mbox{SINR}}_{jk}=\frac{S}{\frac{\sigma^2}{p_d n} + I_{PC} +
I_{MU}},
\end{equation}
\textit{where the desired signal power $(S)$, the power of interference due to PC $(I_{PC})$, and uncorrelated multiuser interference $(I_{MU})$ are respectively given by}
\begin{subequations}\label{eq:S and I}
\begin{align}
\label{eq:S}
  S =& \beta^2 \left( M^\iota \nu_1+(M-1) \alpha_1^2 \nu_2 \right),\\
\label{eq:I_PC}
 I_{PC}=& \beta^2 \alpha_2 \left( \bar L_1 -M^ \frac{\iota}{2} \right) \frac{\left(M^ \frac{\iota}{2} \nu_1+(M-1) \alpha_1 \nu_2 \right)^2}{\left(M^\iota \nu_1+(M-1) \alpha_1^2 \nu_2 \right)},\\
\label{eq:I_MU}
I_{MU}=&\frac{1}{n} I_{MU}' \nonumber \\
=& \frac{\beta d K}{n} \left( M^{\frac{\iota}{2}-1}+(1-\frac{1}{M}) \alpha_1 + \alpha_2 \left(L -1 \right)\right),
\end{align}
\end{subequations}
\textit{where} $\bar L_1=M^ \frac{\iota}{2}+\alpha_2(L / \psi-1)$, $\bar L_2=\alpha_1+\alpha_2(L / \psi-1)$, $\nu_1=p_u \tau_u d / (\sigma^2+p_u \tau_u \bar L_1 \beta d)$, $\nu_2=p_u \tau_u d / (\sigma^2+p_u \tau_u \bar L_2 \beta d)$, \textit{and $I_{MU}'$ is the uncorrelated multiuser interference scaled by $n$.}
\end{Corollary}
\begin{proof}
  See Appendix A.
\end{proof}

From Corollary \ref{Corollary: SINR of simplified model}, we know that $S$ and $I_{PC}$ do not change with the number of each RRH antennas $n$, while $I_{MU}$ and the noise vanish when $n$ grows to infinity.

Assume that the $K$ users achieve a uniform rate $\gamma$ averaged over user locations\footnote{The uniform rate assumption is based on the large-scale fading averaged over different user locations, so we call it uniform rate averaged over different user locations, or simply, average uniform rate.}, solving $p_d$ from \eqref{eq:DE of simplified}, we get the transmit power
\begin{equation}\label{eq:pd}
p_d=\frac{\sigma^2}{n \left( \frac{S}{2^\gamma -1}- I_{PC}\right)-I_{MU}'}.
\end{equation}
\begin{Remark}\label{Remark:mimimum n}
\textit{To achieve user rate $\gamma$, the transmit power $p_d$ should be positive, from \eqref{eq:pd}, we know that the antenna number $n$ must satisfy}
\begin{equation}\label{eq:minimum n}
n> \frac{I_{MU}'}{\frac{S}{2^\gamma-1}-I_{PC}}.
\end{equation}
\end{Remark}

Since the transmit power in \eqref{eq:pd} and the backhaul power are increasing with $\gamma$, the total power consumption $P_{\text{Total}}$ is a function of $\gamma$. With average uniform rate $\gamma$, the cell EE can be expressed as
\begin{equation}\label{eq: EE of simplied model}
  \eta=\frac{\frac{T-\tu}{T} K \gamma}{P_{\text{Total}}(\gamma)},
\end{equation}
with $P_{\text{Total}}(\gamma)=P_{\text{FIX}}+n M P_{\text{RRH}}+\frac{T-\tau_u}{T} \frac{p_d}{\zeta} K + M(P_0+P_{\text{BT}} \frac{T-\tu}{T} K \gamma)$, and $p_d$ is given by \eqref{eq:pd}.

\begin{figure}
\centering
\includegraphics[width=0.5\textwidth]{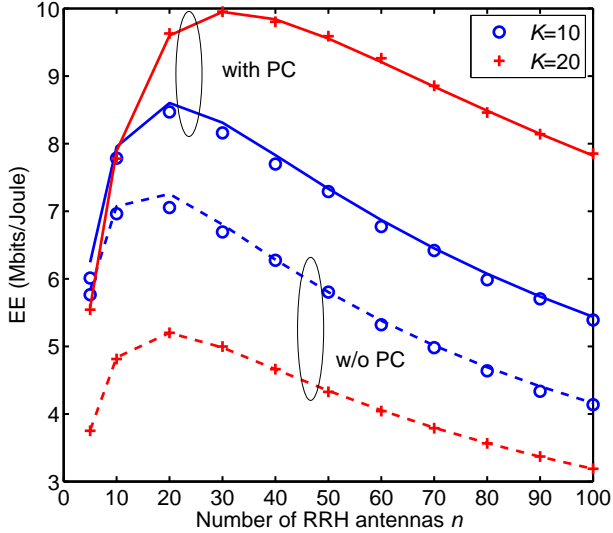}
\caption{Accuracy of asymptotic EE. $p_d=30$ dBm, $M=7$, and $d=1$. The solid curves depict analytical results, while the markers depict simulation results. Two cases are considered: with pilot contamination (denoted as ``with PC") and without pilot contamination (denoted as ``w/o PC").}\label{EE_asymptotic}
\end{figure}

Before we proceed, we verify the accuracy of the derived asymptotic EE at different number of RRH antennas $n$. In Fig.~2, we show the EE when $p_d=30$ dBm, $M=7$, $d=1$, $K=10$ and $20$, respectively. In the case with pilot contamination (denoted as ``with PC"), we set the pilot reuse factor $\psi=1$, and in the case without pilot contamination (denoted as ``w/o PC"), we set $\psi=L$. Other simulation parameters are listed in the beginning of Section V.
It can be observed that the asymptotic results (solid curves) agree with the simulation results (markers) achieved by Monte-Carlo averaging over 1000 channel realizations, even for small number of antennas $n$.
We conclude that the asymptotic EE is accurate even in practical non-asymptotic regimes, and thus can be applied to the optimization problems discussed in the sequel.
\section{Energy Efficiency Optimization}
In this section, we will answer the following questions: For a given uniform rate averaged over different user locations, to maximize the EE, how many antennas should be employed by each RRH? What is the optimal number of users? How many RRHs should be deployed? And what are the impacts on these optimal values due to different parameters, e.g. the channel correlation, the channel gain, the power consumption parameters, and the PC?
\subsection{The Optimal Number of each RRH Antennas $n$}
We first derive and analyze the optimal value of $n$ with fixed $M$ and $K$.
Based on \eqref{eq: EE of simplied model}, the EE optimization problem can be formulated as
\begin{align}\label{eq: EE joint optimization problem}
\max\limits_{n} \quad & \eta=\frac{\frac{T-\tu}{T} K \gamma}{P_{\text{Total}}(\gamma)}, \\
\mbox{s.t.} \quad & \eqref{eq:minimum n}, \ n \in \mathbb Z_+. \nonumber
\end{align}
For a given average uniform rate $\gamma$, the problem can be reduced to
\begin{align}\label{eq: EE equivalent optimization}
\min\limits_{n} \quad & P_{\text{Total}}, \\
\mbox{s.t.} \quad & \eqref{eq:minimum n}, \ n \in \mathbb Z_+. \nonumber
\end{align}
For convenience, we introduce a notation:
\begin{equation}
\lfloor x \rceil_\eta=\begin{cases}
  \lfloor x \rfloor, \ \text{if} \ \eta(\lfloor x \rfloor) > \eta (\lceil x \rceil),\\
\lceil x \rceil, \ \text{otherwise},
\end{cases}
\end{equation}
where $\lfloor x \rfloor$ denotes the largest integer not greater than $x$, and $\lceil x \rceil$ denotes the smallest integer not less than $x$.
\begin{Theorem}\label{Theorem: the optimal n}
\textit{For a given uniform rate $\gamma$ averaged over different user locations, the optimal number of RRH antennas that maximizes the EE is}
\begin{equation}\label{eq: optimal n}
n^\star= \left\lfloor  \sqrt{\frac{\frac{T-\tau_u}{T \zeta} \sigma^2 K}{ \left( \frac{S}{2^\gamma -1}-I_{PC} \right) M P_{\text{RRH}}}}+ \frac{I_{MU}'}{\frac{S}{2^\gamma -1}-I_{PC}} \right\rceil_\eta.
\end{equation}
\end{Theorem}
\begin{proof}
See Appendix B.
\end{proof}

From Theorem \ref{Theorem: the optimal n}, some insights on how $n^\star$ is affected by other system parameters can be obtained, the results are described in the following remark.
\begin{Remark}\label{Remark: insights of optimal n}
\textit{From Theorem \ref{Theorem: the optimal n}, the following observations can be made:}
\begin{enumerate}
\item When $K$ increases, the scaled multi-user interference $I_{MU}'$ increases, and $n^\star$ increases with $K$ accordingly.
\item When $P_{\text{RRH}}$ decreases, $n^\star$ increases. That is to say, using lower power consuming hardware components to reduce $P_{\text{RRH}}$, $n^\star$ will increase.
\item When the noise is comparably negligible ($\sigma^2 \ll p_u \tau_u \bar L_1 \beta d$), $n^\star$ is an increasing function of $d$. A large value of $d$ means an environment with insufficient scattering, in this case, more antennas are required to achieve the optimum EE.
\item When the noise is comparably negligible ($\sigma^2 \ll p_u \tau_u \bar L_1 \beta d$), as the cell size increases, or $\beta$ decreases, $n^\star$ will increase.
\item When the pilot reuse factor $\psi$ decreases, or the PC becomes more serious, $n^\star$ will increase.
\end{enumerate}
\end{Remark}
\begin{proof}
1) and 2) can be observed from \eqref{eq:I_MU} and \eqref{eq: optimal n} directly. When the noise is negligible, i.e., $\sigma^2 \ll p_u \tau_u \bar L_1 \beta d$, we have $\nu_1 \approx 1/(\bar L_1 \beta)$, $\nu_2 \approx 1/(\bar L_2 \beta)$. Substituting $\nu_1$ and $\nu_2$ into \eqref{eq:S and I}, it can be known that $S$, $I_{PC}$ and $I_{MU}'$ depend linearly on $\beta$, and both $S$ and $I_{PC}$ are independent of $d$, while $I_{MU}'$ increases with $d$. Thus, $n^\star$ increases with $d$, and decreases with $\beta$, which are summarized in 3) and 4). When $\psi$ decreases, $I_{PC}$ increases, and more antennas should be deployed to achieve the maximal EE.
\end{proof}

The above observations can also be explained as follows:

With more users, the multi-user interference increases, hence more antennas are required to achieve the target rate $\gamma$.
When $P_{\text{RRH}}$ becomes larger, more power is required to run each RRH antenna, in this case, the transmit power $p_d$ is small when compared to the power consumed for running the antennas, and thus using more antennas may increase the total power consumption and decrease the EE. However, if $P_{\text{RRH}}$ is small and fixed, the running power of antennas is smaller than $p_d$. When $d$ is larger or the average channel gain $\beta$ is smaller, increasing the number of antennas will improve the array gain to reduce $p_d$. In such a scenario, it is optimal to equip more antennas to reduce the total power consumption and improve the EE. When $\psi$ decreases, the pilot sequences will be reused in more cells, the interference due to pilot contamination will increase, and hence a large array gain is needed to reduce the required transmit power $p_d$ and then improve the EE.
\begin{Corollary}\label{Corollary: the optimal n without PC}
\textit{The optimal $n^\star$ is lower bounded when there is no PC ($I_{PC}=0$), which is given by \eqref{eq: the lower bound of optimal n} , shown at the top of next page.}
\begin{figure*}[!t]
\begin{equation}\label{eq: the lower bound of optimal n}
n^\star= \left\lfloor  \sqrt{\frac{\frac{T-\tau_u}{T \zeta} \sigma^2 K}{ \frac{\beta\left(M^{\frac{\iota}{2}} + \left(M-1\right) \alpha_1 \right)}{2^\gamma -1} M P_{\text{RRH}}}}+ \frac{d K \left(M^{\frac{\iota}{2}-1}+(1-\frac{1}{M})\alpha_1 +\left(L-1\right)\alpha_2\right)}{\frac{\left(M^ \frac{\iota}{2} + \left(M-1\right) \alpha_1 \right)}{2^\gamma -1}} \right\rceil_\eta.
\end{equation}
\end{figure*}
\end{Corollary}
\begin{proof}
From Remark \ref{Remark: insights of optimal n}, $n^\star$ is decreasing with $\psi$, in the case without PC, $\psi=L$, $\bar L_1=M^ \frac{\iota}{2}$, $\bar L_2=\alpha_1$, and $I_{PC}=0$. Substituting these results into \eqref{eq: optimal n} yields  Corollary \ref{Corollary: the optimal n without PC}.
\end{proof}
\begin{Remark}
From Corollary \ref{Corollary: the optimal n without PC}, we can know that when the inter-cell interference factor $\alpha_2$ increases, more antennas are required to achieve the maximum EE.
\end{Remark}

\subsection{The Optimal Number of Users $K$}
With more users in each cell, the sum rate will increase accordingly, but to satisfy the given average uniform rate, the transmit power is proportional to the number of users $K$ as well, thereby there exists an optimal value of $K$ to maximize the EE. We now investigate the optimal number of users when other parameters are given.
The problem is formulated as
\begin{align}\label{eq: optimize K}
\max\limits_{K} \quad & \eta=\frac{\frac{T-\tu}{T} K \gamma}{P_{\text{Total}}(\gamma)}, \\
\mbox{s.t.} \quad & p_d>0, \ K \in \mathbb Z_+. \nonumber
\end{align}

Plugging $\tau_u=\psi K$ and \eqref{eq:pd} into \eqref{eq: EE of simplied model}, the EE is given by \eqref{eq: EE expression of K}, shown at the top of next page.
\begin{figure*}[!t]
\begin{equation}\label{eq: EE expression of K}
  \eta=\frac{\frac{T-\psi K}{T} K \gamma}{P_{\text{FIX}}+n M P_{\text{RRH}}+\frac{T-\psi K}{T} \frac{\sigma^2/\zeta}{n \left( \frac{S}{2^\gamma -1}- I_{PC}\right)-I_{MU}'}K +M(P_0+P_{\text{BT}} \frac{T-\psi K}{T} K \gamma)}.
\end{equation}
\hrulefill
\end{figure*}

When the noise is comparably negligible, $\nu_1 \approx 1/(\bar L_1 \beta)$, $\nu_2 \approx 1/(\bar L_2 \beta)$. Then, in \eqref{eq: EE expression of K}, the scaled multiuser interference $I_{MU}'$ is the function of $K$, while the desired signal power $S$ and the power of PC interference $I_{PC}$ are independent of $K$. For notation convenience, we rewrite $I_{MU}'$ in the form
\begin{equation}\label{eq: rewrite IMU}
I_{MU}'=\beta d K \xi,
\end{equation}
where $\xi=M^{\frac{\iota}{2}-1}+(1-\frac{1}{M}) \alpha_1 + \alpha_2 \left(L -1 \right)$.
\begin{Theorem}\label{Theorem:optimal K}
\textit{For a given uniform rate $\gamma$ averaged over different user locations, when the noise is comparably negligible, the optimal number of users that maximizes the EE is}
\begin{equation}\label{eq: the optimal K}
  K^\star=\left\lfloor K^\circ \right\rceil_\eta,
\end{equation}
\textit{where $K^\circ$ is the root in the range $(0, \min\{\frac{T}{\psi},  \frac{\mu_1} {d \beta \xi}\})$ of the following equation}
\begin{equation}\label{eq: the equation of K}
\mu_2 (2 K \psi -T) \left( \mu_1-d \beta \xi K \right)^2 + \frac{\sigma^2}{\zeta \gamma} d \beta \xi \big( (T-K \psi) K \big)^2 =0,
\end{equation}
\textit{with} $\mu_1=n \left( \frac{S}{2^ \gamma -1}-I_{PC}\right)$ and $\mu_2= \frac{T}{\gamma} (P_{\text{FIX}}+n M P_{\text{RRH}} + M P_0)$.

\begin{proof}
See Appendix C.
\end{proof}
\end{Theorem}

Theorem~\ref{Theorem:optimal K} shows that $K^\star$ is a root of the quartic equation given by \eqref{eq: the equation of K}. The closed-form root expressions of a quartic equation can be found in \cite{Hungerford-96}. Due to the lengthy and complexity of these expressions, we can use a numerical algorithm, e.g., bisection method, to find the root in the range $(0, \min\{\frac{T}{\psi},  \frac{\mu_1} {d \beta \xi}\})$. Moreover, from \eqref{eq: the equation of K} we know that $K^\star$ is related to $\mu_2$, that is, $K^\star$ also depends on the terms of power consumption that are independent of $K$, including $P_{\text{FIX}}$, $P_{\text{RRH}}$, and $P_0$.

\subsection{The Optimal Number of RRHs $M$}
In the massive DAS we considered, the number of RRHs will influence the EE performance. On the one hand, the channel gain (or the distance) between ${\sf UE}_{j,k}$ and ${\sf RRH}_{j,\bar m_{jk}}$ is changing with the number of RRHs, on the other hand, the power consumption of backhaul increases with the number of RRHs. Given other system parameters, with a average uniform rate, the optimal number of RRHs $M$ for EE maximization problem can be formulated as
\begin{align}\label{eq: EE equivalent optimization of M}
\min\limits_{M} \quad & P_{\text{Total}}, \\
\mbox{s.t.} \quad & p_d(M)>0, \ M \in \mathbb Z_+. \nonumber
\end{align}

Due to the complex expression of $M$ in $\eta$, the closed-form of $M^\star$ is not allowed. However, $M^\star$ can be obtained efficiently with a one-dimensional search over the candidate set $\{1,2,\dots, M_{\max}\}$, i.e.
\begin{align}
 M^\star =&\argmin\limits_{\substack{M \in \{1,2,\dots, M_{\max}\}}} \quad  P_{\text{Total}},\\
 & \mbox{s.t.} \quad  p_d(M)>0, \nonumber
\end{align}
where $M_{\max}$ is a predefined value\footnote{We will see in simulations that EE first increases and then decreases with $M$. Thus, $M_{\max}$ can be determined from the behavior of EE. Moreover, we observe that the optimal $M$ is increasing with the number of users $K$, hence, $M_{\max}$ could be set as scaled with $K$.}. As shown in \eqref{eq:pd} and \eqref{eq: EE of simplied model}, $p_d$ and $P_{\text{Total}}$ are independent of instantaneous CSI, and hence $M^\star$ is independent of instantaneous CSI. $P_{\text{Total}}$ is related to $n$, $K$, $\gamma$, $\iota$, $\beta$, and the power consumption parameters.
Given these system parameters, $M^\star$ can be obtained by searching over \{$1, 2,\cdots, M_{\max}$\} only once, and it remains the same as long as these parameters unchanged.
\section{Simulation Results}
In this section, we conduct numerical simulations to confirm our analytical results. We set $L=7$, and the large-scale fading factors in \eqref{eq: interference factor} are chosen as follows. We consider the 7-RRH massive DAS as illustrated in Fig.~\ref{systemmodel}. In each cell, $K=10$ users are located uniformly at random. We take the 10 users in the center cell (indexed by cell 1) as samples. Let $\bar \beta_0$ be the average of the large-scale fading factors $\beta_{1\bar m_{1k}1k}$ over the 10 users, $\bar \beta_1$ be the average of $\beta_{1m1k} \ (m \neq \bar m_{1k})$ over the $M-1$ RRHs and the 10 users, and $\bar \beta_2$ be the average $\beta_{jm1k} \ (j\neq 1)$ over the RRHs in other six cells and the 10 users.
We generate 1000 random user locations to calculate $\Ex\left\{\bar \beta_0\right\}$, $\Ex\left\{ \bar \beta_1\right\}$, and $\Ex\left\{ \bar \beta_2\right\}$. Base on \eqref{eq: interference factor}, we compute the average channel gain $\beta$, the interference factor $\alpha_1$, and $\alpha_2$ as $\Ex\left\{\bar \beta_0\right\}/M^\frac{\iota}{2}$, $\Ex\left\{\bar \beta_1\right\}/\beta$, and $\Ex\left\{\bar \beta_2\right\}/\beta$, respectively. By setting the cell radius $R_c$ be 2 km, and the path-loss exponent $\iota$ be $2.5$, we obtain $\beta=2.24 \times 10^{-8}$, $\alpha_1=0.54$, and $\alpha_2=0.075$.

Other simulation parameters are  defined in Table~\ref{simulation parameters} \cite{Onireti-13TCOM,EmilB-15TWC}. Unless otherwise stated, we keep these parameters in the following simulations. The detailed discussions are as follows.

\begin{table}
\centering
\caption{Simulation Parameters} \label{simulation parameters}
\begin{tabular}{|c|c|}
\hline Parameter & Value \\
\hline Amplifier efficiency : $\zeta$ & 0.4 \\
\hline Coherence interval : T & 196   \\
\hline  System bandwidth: $B$ & 20 MHz  \\
\hline  Fixed backhaul power: $P_0$ & 0.825 W \\
\hline  Traffic dependent backhaul power:${P_\text{BT}}$   & 0.25W/(Gbits/s)\\
\hline  Fixed system power: ${P_\text{FIX}}$   & 9 W \\
\hline  Power of each antenna at RRH: ${P_\text{RRH}}$ & 0.2 W \\
\hline Total noise power: $N_0 B$ & $-40$ dBm \\
\hline
\end{tabular}
\end{table}

\subsection{Impact of channel correlation and channel gain on the maximal EE and the optimal $n$}
\begin{figure}
\centering
\includegraphics[width=0.5\textwidth]{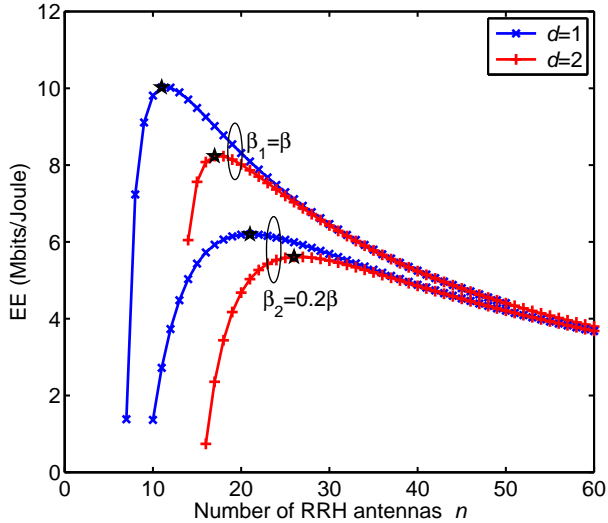}
\caption{Impact of channel correlation $d$ and average channel gain $\beta$ on the maximal EE and $n^\star$. $\psi=1$, $M=7$, $K=10$, and $\gamma=2$. $n^\star$ increases with $d$ and decrease with $\beta$, while the maximal EE decreases with $d$ and increases with $\beta$.}\label{EE_dandbeta}
\end{figure}

The EE achieved by different number of RRH antennas $n$ for different values of $d$ and $\beta$ when the pilot reuse factor $\psi=1$ are presented in Fig.~\ref{EE_dandbeta}. An average uniform rate $\gamma$ of 2 bit/s/Hz and a fixed number of RRHs $M=7$ are assumed. From the simulation result we note that when $\beta_1=2.24 \times 10^{-8}$, $n^\star=11$ and $n^\star=17$ are optimal to maximize the EE for $d=1$ and $d=2$, respectively. When $\beta_2=0.2 \beta_1$, $n^\star=21$ and $n^\star=26$ are optimal for $d=1$ and $d=2$, respectively. These optimal values agree with the results from Theorem \ref{Theorem: the optimal n} (marked with $\star$). From the curves, we conclude that when the channel gain $\beta$ is fixed, as compared to the scenario without channel correlation ($d=1$), with channel correlation ($d=2$), the optimal number of antennas to achieve the maximal EE will be larger, but the achieved maximal EE is lower, since the power to run the total antennas increases. Comparing the two sets of curves of $\beta_1=2.24 \times 10^{-8}$ and $\beta_2=0.2 \beta_1$, when $\beta$ decreases, $n^\star$ increases, and a higher average channel gain results in a higher maximal EE. These insights are consistent with 3) and 4) of Remark \ref{Remark: insights of optimal n}.

\subsection{Impact of PC and the power of each RRH antenna $P_{\text{RRH}}$ on the maximal EE and the optimal $n$}
\begin{figure}
\centering
\includegraphics[width=0.5\textwidth]{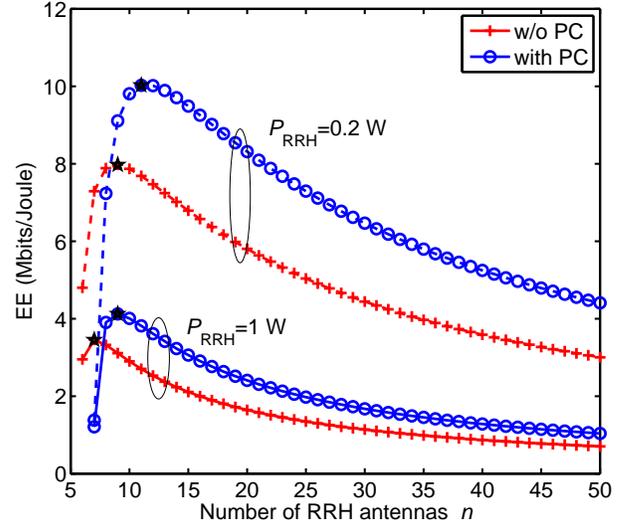}
\caption{Impact of PC and $P_{\text{RRH}}$ on the maximal EE and $n^\star$. $d=1$, $M=7$, $K=10$, and $\gamma=2$. As compared to the scenario without PC, $n^\star$ is larger for the case with PC. When the running power of each RRH antenna $P_{\text{RRH}}$ is lower, more antennas are required to achieve a higher maximal EE.}\label{EE_PBS}
\end{figure}

The impact of PC and $P_{\text{RRH}}$ on $n^\star$ and the maximal EE are investigated in Fig.~\ref{EE_PBS}. Here, we compare the EE of massive DAS with parameter in the year 2011 and the predicated value in 2020, which are respectively  $P_{\text{RRH}}= 1$ W and $P_{\text{RRH}}= 0.2$ W \cite{Kumar-11WCVT,Desset-14GreenComm}. As pointed out in Corollary \ref{Corollary: the optimal n without PC}, $n^\star$ will be larger when there exists PC. For the impact of $P_{\text{RRH}}$, we can see that when $P_{\text{RRH}}=1$ W, the maximal EE is degraded severely, and $n^\star$ is almost the same as the minimum number of antennas required to achieve the average uniform rate $\gamma=2$. Therefore, if the hardware components of RRH antennas are power inefficient, it is not wise to deploy a large number of antennas from the viewpoint of EE.

\subsection{Impact of inter-cell interference on the maximal EE and the optimal $n$}
\begin{figure}
\centering
\includegraphics[width=0.5\textwidth]{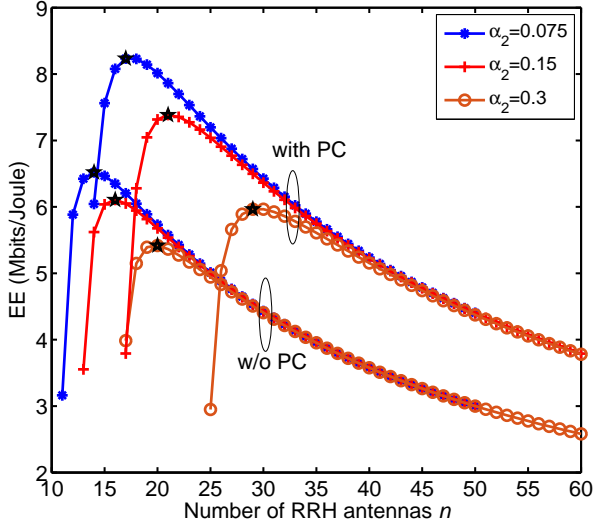}
\caption{Energy efficiency versus the number of RRH antennas for different values of inter-cell interference factor $\alpha_2$ with and without PC. $d=2$, $M=7$, $K=10$, and $\gamma=2$. $n^\star$ increases with $\alpha_2$. As compared to the case without PC, the impact of $\alpha_2$ on $n ^\star$ is more obvious for the case with PC.}\label{EE_alpha2}
\end{figure}

Fig.~\ref{EE_alpha2} shows the set of EE values with and without PC for different values of inter-cell interference factor $\alpha_2$. For the case with PC, when $\alpha_2$ is set to be 0.075, 0.15, and 0.3, $n^\star$ increases from 17 to 21 and 29, respectively. However, the increase of $n^\star$ for the case without PC is not obvious when compared to that with PC scenario. This happens because when $\alpha_2$ increases, both the interference due to PC and the uncorrelated multi-user interference increases, and the effect of PC becomes more serious when $n$ becomes larger.

\subsection{The trade-off between EE and average uniform rate $\gamma$}
\begin{figure}
\centering
\includegraphics[width=0.5\textwidth]{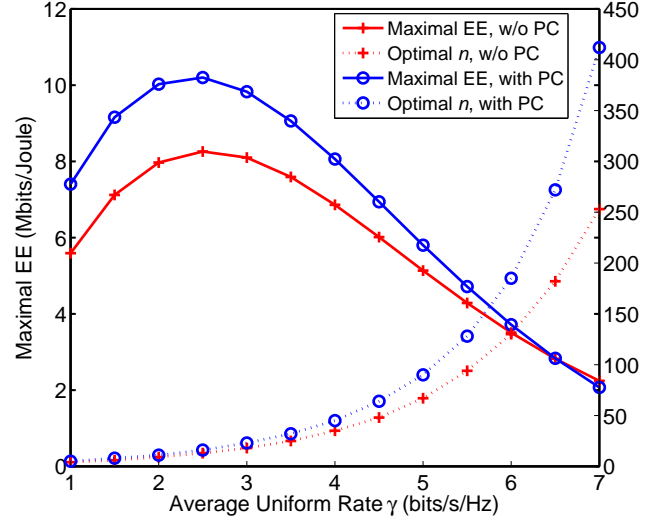}
\caption{Maximal energy efficiency and the corresponding optimal $n$ versus the average uniform rate $\gamma$ with and without PC. $d=1$, $K=10$, and $M=7$.}\label{EE_gamma}
\end{figure}

In Fig.~\ref{EE_gamma}, both the maximal EE and the corresponding $n^\star$ are displayed as a function of $\gamma$ when $M$ is fixed to 7.
We observe that when $\gamma$ is not large, the maximal EE and $\gamma$ can simultaneously increase, but when $\gamma$ is larger than a value, the maximal EE decreases inversely. This is because when $\gamma$ is increasing, the required number of antennas $n$ increases accordingly. And when the proportion of the increase of the user rate is less than that of the increased power to run the RRH antennas, the EE decreases. We also note that to achieve the maximal EE, $n^\star$ increases faster with $\gamma$ for the case with PC, and thus the EE also decreases faster.

\subsection{Impact of PC and channel correlation on the maximal EE and the optimal $K$}
\begin{figure}
\centering
\includegraphics[width=0.5\textwidth]{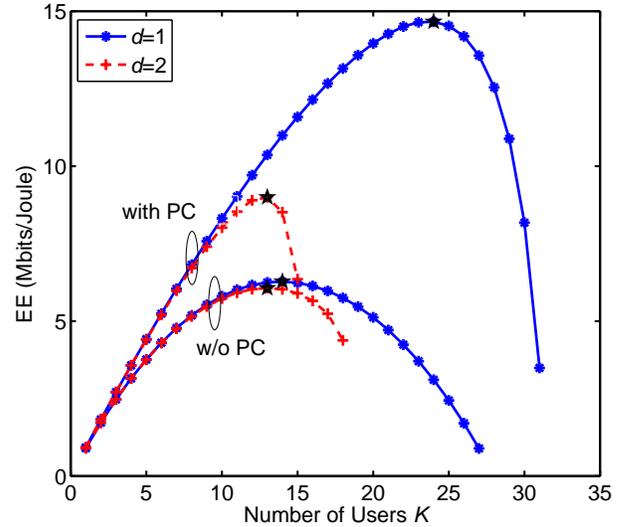}
\caption{Impact of PC and channel correlation on the maximal EE and $K^\star$. $n=20$, $M=7$, and $\gamma=2$. When the channel correlation is absent ($d=1$), more users can be served to maximize the EE.}\label{EE_K}
\end{figure}

Fig.~\ref{EE_K} illustrates the EE versus the number of users for $d=1$, $d=2$, with and without PC, respectively. $M=7$ RRHs are deployed in each cell, and the antenna number of each RRH is fixed at 20. The figure shows that when $d=1$, the maximal values of EE for the case with and without PC are obtained at $K=24$ and $K=14$, respectively. When $d=2$, the maximal EE are obtained at $K=13$, which are consistent with the results of using bisection method in Theorem \ref{Theorem:optimal K} (marked with $\star$). When $d=1$, the optimal $K$ to maximize EE for the case without PC is less than that with PC, this is because for the case without PC, if a larger number of users are served, in per coherence interval, the length of uplink pilot sequence $\tau=KL$ will be large, and less symbols can be used for downlink data transmission, which degrades the SE and EE.

\subsection{Impact of $K$ on the maximal EE and the optimal $M$}
\begin{figure}
\centering
\includegraphics[width=0.5\textwidth]{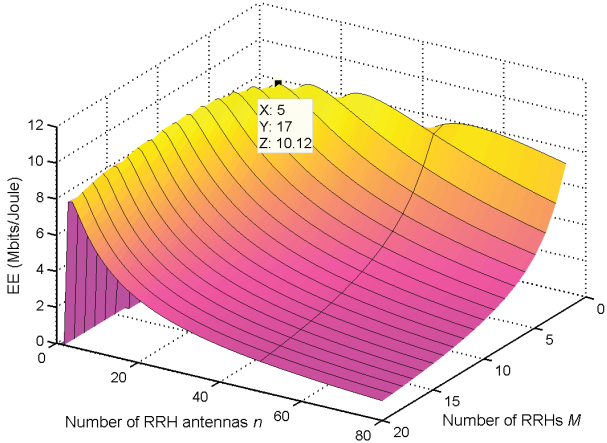}
\caption{EE with the numbers of RRHs $M$ and RRH antennas $n$. $K=10$, $\psi=1$, $d=1$ and $\gamma=2$. The optimal EE 10.12 Mbits/J is obtained at $(M,n)=(5,17)$.}\label{EE_M_3D}
\end{figure}

\begin{figure}
\centering
\includegraphics[width=0.5\textwidth]{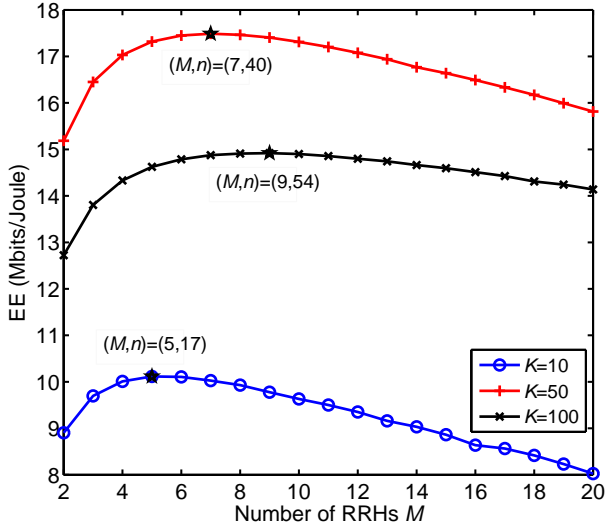}
\caption{Impact of $K$ on the maximal EE and $M^\star$. $\psi=1$, $d=1$ and $\gamma=2$. With more users, more antennas and RRHs should be deployed to maximize the EE.}\label{EE_M}
\end{figure}

Fig. \ref{EE_M_3D} shows the achievable EE with different numbers of RRHs $M$ and RRH antennas $n$ when $K=10$. The figure shows that the optimal EE 10.12 Mbits/J is achieved at $(M,n)=(5,17)$. We then consider the relationship between EE and $(M,n)$ for two other numbers of users, i.e., medium users ($K=50$) and a large number of users ($K=100$). The 3D graphs for this two cases are similar to Fig. \ref{EE_M_3D} and are not shown here. The optimal EE versus the number of RRHs $M$ for the three cases of users are presented in Fig.~\ref{EE_M}. Each point uses the EE-optimal value of $n$. The optimal EE are obtained at $(M,n)=(5,17)$, $(7,40)$, and $(9,54)$ for $K=$10, 50, and 100, respectively. We notice that with more users, more antennas and RRHs should be deployed to maximize the EE.

\subsection{EE comparison between massive DAS and massive CAS}
\begin{figure}
\centering
\includegraphics[width=0.5\textwidth]{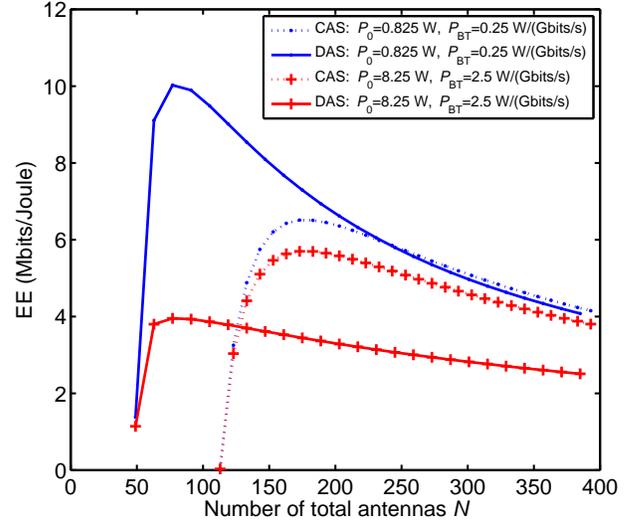}
\caption{EE comparison between DAS ($M=7$) and CAS ($M=1$) under different power consumption of backhaul. $K=10$, $\psi=1$, $d=1$, and $\gamma=2$.}\label{EE_Compare}
\end{figure}

Finally, Fig.~\ref{EE_Compare} shows the EE comparison between massive DAS ($M=7$) and massive CAS ($M=1$) under different consumption of backhauling powers. The solid lines indicate the EE performance of massive DAS, and the dotted lines indicate the EE of massive CAS. As defined in Section II, the backhauling power is modeled as $M (P_0 + R B P_{\text{BT}})$. We first set $P_0=0.825$ W, $P_{\text{BT}}=0.25$ W/(Gbits/s), and then we change these parameters to $P_0=8.25$ W, $P_{\text{BT}}=2.5$ W/(Gbits/s). We observe that when $P_0=0.825$ W, $P_{\text{BT}}=0.25$ W/(Gbits/s), massive DAS is more energy efficient than massive CAS, and vice versa as in the case of $P_0=8.25$ W, $P_{\text{BT}}=2.5$ W/(Gbits/s). The reason is that in DAS, the average distance between the RRH and users is decreased, and thus the transmit power is less. It is also shown that to achieve the maximal EE, the optimal number of total antennas $N=mM$ of DAS is less than that of CAS, so the power to run the total antennas ($NP_{\text{RRH}}$) is decreased. In DAS, more power is consumed for backhauling, if the backhaul links are power efficient, massive DAS can achieve higher EE than CAS. However, if the backhauling power is large, massive CAS will be more energy efficient than massive DAS, because a significant increase of the total power consumption is used for backhauling, which decreases the EE of massive DAS.

\section{Conclusion}
In this paper, under a realistic power consumption model, we have investigated the problem of maximizing the EE of a downlink multi-cell massive DAS, with respect to the number of RRH antennas $n$, the number of RRHs $M$, and the number of served users $K$.
Our study provided an efficient tool to help the system designer in deciding the optimal $n$, $M$, and $K$ that achieving the optimal EE.
Simulation results validated our analysis, and demonstrated that the DAS is always more energy efficient than CAS, unless the backhauling power is large. In addition, more RRHs and antennas should be used to achieve the optimal EE when the number of users is increased. While having more antennas may lead to higher PC, we show that for a system with PC, to achieve the optimal EE, more antennas are needed when compared to that of the system without PC.

\section*{Appendix}
\subsection{Proof of Corollary \ref{Corollary: SINR of simplified model}}
Under the simplified channel model, we have
\begin{equation}\label{eq: R}
 \qR_{lmjk}=\beta_{lmjk} \frac{n}{P}\qA \qA^H.
\end{equation}
Based on \eqref{eq: R}, we have
\begin{align}
\qQ_{lmjk}&= \big( \frac{\sigma^2}{p_u \tu} \qI_{n} + \sum\limits_{j \in \calL_j} \qR_{lmjk}\big)^{-1} \nonumber \\
&=\begin{cases}
\big( \frac{\sigma^2}{p_u \tu} \qI_{n} + \bar L_1\beta \frac{n}{P}\qA \qA^H \big)^{-1}, &\text{if} \quad m=\bar m_{lk}, \\
\big( \frac{\sigma^2}{p_u \tu} \qI_{n} + \bar L_2\beta \frac{n}{P}\qA \qA^H \big)^{-1}, &\text{if} \quad m \neq \bar m_{lk},
\end{cases}
\end{align}
where $\bar L_1=M^ \frac{\iota}{2}+\alpha_2(L/\psi-1)$, and $\bar L_2=\alpha_1+\alpha_2(L/\psi-1)$.

Using matrix inversion lemma $\qP(\qI+\qW \qP)^{-1}=(\qI+ \qP \qW)^{-1}\qP$, and the fact that $\qA^H \qA=\qI_P$, when $l \neq j$, we have
\begin{align}
 \qPhi_{lmjk}&=\qR_{lmlk} \qQ_{lmjk} \qR_{lmjk} \nonumber \\
&=\begin{cases}
M^ \frac{\iota}{2} \alpha_2 \beta^2 d \nu_1 \qA \qA^H, &\text{if} \quad m=\bar m_{lk}, \\
\alpha_1 \alpha_2 \beta^2 d \nu_2 \qA \qA^H, &\text{if} \quad m \neq \bar m_{lk},
\end{cases}
\end{align}
with $\nu_1=p_u \tau_u d / (\sigma^2+p_u \tau_u \bar L_1 \beta d)$, and $\nu_2=p_u \tau_u d / (\sigma^2+p_u \tau_u \bar L_2 \beta d)$.

Similarly, when $l=j$, we have
\begin{align}
 \qPhi_{jmjk}&=\qR_{jmjk} \qQ_{jmjk} \qR_{jmjk} \nonumber \\
&=\begin{cases}
M^\iota \beta^2 d \nu_1 \qA \qA^H, &\text{if} \quad m=\bar m_{jk}, \\
\alpha_1^2 \beta^2 d \nu_2 \qA \qA^H, &\text{if} \quad m \neq \bar m_{jk}.
\end{cases}
\end{align}
Since $\tr \{\qA \qA^H\}=\tr \{ \qA^H \qA\}=P$, the power of the desired signal can be derived as
\begin{align}
  S&=\bar \lambda_j \left( \frac{1}{n} \sum\limits_{m=1}^M \tr \qPhi_{jmjk}\right)^2 \nonumber \\
  &=\frac{1}{n} \sum\limits_{m=1}^M \tr \qPhi_{jmjk} \nonumber \\
  &=\beta^2 \left(M^\iota \nu_1+(M-1) \alpha_1^2 \nu_2 \right).
\end{align}
 The power of interference due to PC, and multiuser interference can be derived as follows.
\begin{align}
 I_{PC}&=\sum\limits_{l \neq j} \bar \lambda_l \left| \frac{1}{n} \sum\limits_{m=1}^M \tr \qPhi_{lmjk}\right|^2 \nonumber \\
 &=\beta^2 \alpha_2 \left( \bar L_1 -M^ \frac{\iota}{2} \right) \frac{\left(M^ \frac{\iota}{2}\nu_1+(M-1) \alpha_1 \nu_2 \right)^2}{\left(M^\iota \nu_1+(M-1) \alpha_1^2 \nu_2 \right)}.
\end{align}
\begin{align}
I_{MU}=& \frac{\beta d K}{n} \left( M^{\frac{\iota}{2}-1}+(1-\frac{1}{M}) \alpha_1 + \alpha_2 \left(L -1 \right)\right).
\end{align}

\subsection{Proof of Theorem \ref{Theorem: the optimal n}}
Plugging \eqref{eq:pd} into the expression of $P_{\text{Total}}$ in \eqref{eq: EE of simplied model}, the optimization problem \eqref{eq: EE equivalent optimization} can be expressed as
\begin{align}\label{eq: EE optimization problem fn}
\min\limits_{n} \quad & f(n), \\
\mbox{s.t.} \quad & \eqref{eq:minimum n}, \ n \in \mathbb Z_+,\nonumber
\end{align}
where $f(n)=n M P_{\text{RRH}}+ K \frac{T-\tau_u}{T\zeta} \frac{\sigma^2}{n \left( \frac{S}{2^\gamma -1}- I_{PC}\right)-I_{MU}'}$.

When $P_{\text{RRH}}$ and $p_d$ are positive, the two items at the right-hand side of $f(n)$ are both positive. From mean value equalities, $Ax+\frac{B}{x-C} \geq AC+2 \sqrt{AB}$, if $A$, $B$ and $x-C$ are positive, and the equality holds only when $x=C+\sqrt{\frac{B}{A}}$. Based on this, the optimal $n^\circ$ that minimize $f(n)$ is found to be
\begin{equation}
n^\circ=\sqrt{\frac{\frac{T-\tau_u}{T \zeta} \sigma^2 K}{ \left( \frac{S}{2^\gamma -1}-I_{PC} \right) M P_{\text{RRH}}}}+ \frac{I_{MU}'}{\frac{S}{2^\gamma -1}-I_{PC}}.
\end{equation}
It can be easily found that the first-order derivative of $f(n)$ is increasing for $n\in (n^\circ, \infty) $, and decreasing for $n \in (\frac{I_{MU}'}{\frac{S}{2^\gamma -1}-I_{PC}}, n^\circ]$. Therefore, $f(n)$ is a strictly quasi-convex function. Since the number of antennas is a positive integer, the quasi-convexity of $f(n)$ implies that $n^\star$ is the closest integer smaller or larger than $n^ \circ$, which is determined by comparing the EE achieved by the two closest integers. Thus, the proof is completed.

\subsection{Proof of Theorem \ref{Theorem:optimal K}}
We first consider the first-order derivative of $\frac{1}{\eta}$.
\begin{equation}
\frac{\partial}{\partial K}\frac{1}{\eta}=\frac{z(K)}{\big( (T-K \psi) K \big)^2 \left( \mu_1-d \beta K \xi \right)^2},
\end{equation}
where
\begin{equation}\label{eq: the derivate of zk}
z(K)=\mu_2 (2 K \psi -T) \left( \mu_1-d \beta \xi K \right)^2 + \frac{\sigma^2}{\zeta \gamma} d \beta \xi \big( (T-K \psi) K \big)^2.
\end{equation}
Since the length of the pilot $\psi K <T$ and the transmit power $p_d > 0$, $K$ should satisfy the constraint as $0< K < \min\{\frac{T}{\psi},  \frac{\mu_1} {d \beta \xi}\}$. From \eqref{eq: the derivate of zk}, we know that the sign of $\frac{\partial}{\partial K}\frac{1}{\eta}$ is the same as that of $z(K)$, and thus we consider $z(K)$ to characterize the shape of $1/\eta$. When $K \rightarrow 0$, $z(K)$ approaches to a negative value as
\begin{equation}
\lim_{K \rightarrow 0} z(K) = -\mu_2 T \mu_1^2.
\end{equation}
If $\frac{T}{\psi} < \frac{\mu_1} {d \beta \xi}$, when $K \rightarrow \frac{T}{\psi}$, $z(K)$ approaches to a positive value as

\begin{equation}
\lim_{K \rightarrow \frac{T}{\psi}} z(K) = \mu_2 T \left( \mu_1-d \beta \xi  \frac{T}{\psi} \right)^2.
\end{equation}

Similarly, if $ \frac{\mu_1} {d \beta \xi} < \frac{T}{\psi}$, when $K \rightarrow \frac{\mu_1} {d \beta \xi}$, $z(K)$ also approaches to a positive value.  By calculating, the first-order derivative of $z(K)$ is positive, which implies that there is a unique $K^\circ$ such that $z(K^\circ)=0$. Since the sign of $z(K)$ is equal to that of $\frac{\partial}{\partial K}\frac{1}{\eta}$, we know that $1/ \eta$ is decreasing for $K \in (0, K^\circ)$ and increasing for $K \in (K^\circ, \min\{\frac{T}{\psi},  \frac{\mu_1} {d \beta \xi}\})$. Therefore, $1/ \eta$ is quasi-convex in the range $[0,\min\{\frac{T}{\psi}, \frac{\mu_1} {d \beta \xi}\}]$, and get the minimum value when $K=K^\circ$, or $\eta$ is maximal when $K=K^\circ$, which yields the result of Theorem \ref{Theorem:optimal K}.
\bibliographystyle{IEEEtran}

\begin{thebibliography}{10}
\providecommand{\url}[1]{#1}
\csname url@rmstyle\endcsname
\providecommand{\newblock}{\relax}
\providecommand{\bibinfo}[2]{#2}
\providecommand\BIBentrySTDinterwordspacing{\spaceskip=0pt\relax}
\providecommand\BIBentryALTinterwordstretchfactor{4}
\providecommand\BIBentryALTinterwordspacing{\spaceskip=\fontdimen2\font plus
\BIBentryALTinterwordstretchfactor\fontdimen3\font minus
  \fontdimen4\font\relax}
\providecommand\BIBforeignlanguage[2]{{%
\expandafter\ifx\csname l@#1\endcsname\relax
\typeout{** WARNING: IEEEtran.bst: No hyphenation pattern has been}%
\typeout{** loaded for the language `#1'. Using the pattern for}%
\typeout{** the default language instead.}%
\else
\language=\csname l@#1\endcsname
\fi
#2}}

\bibitem{Tombaz-11WCM}
{S. Tombaz, A. Vastberg, and J. Zander}, ``{Energy-and cost-efficient
  ultra-high-capacity wireless access},'' \emph{IEEE Wireless Commun. Mag.},
  vol.~18, no.~5, pp. 18--24, Oct. 2011.

\bibitem{Andrews-14JSAC}
{J. G. Andrews, S. Buzzi, W. Choi, S. V. Hanly, A. Lozano, A. C. K. Soong, and
  J. C. Zhang}, ``{What will 5G be?}'' \emph{IEEE J. Sel. Areas Commun.},
  vol.~32, no.~6, pp. 1065--1082, June 2014.

\bibitem{Marzetta-10TWC}
{T. L. Marzetta}, ``{Noncooperative cellular wireless with unlimited numbers of
  base station antennas},'' \emph{IEEE Trans. Wireless Commun.}, vol.~9,
  no.~11, pp. 3590--3600, Nov. 2010.

\bibitem{Rusek-13SPM}
{F. Rusek, D. Persson, B. K. Lau, E. G. Larsson, T. L. Marzetta, O. Edfors, and
  F. Tufvesson}, ``{Scaling up MIMO: Opportunities and challenges with very
  large arrays},'' \emph{IEEE Sig. Proc. Mag.}, vol.~30, no.~1, pp. 40--60,
  Jan. 2013.

\bibitem{Suraweera-13ICC}
{H. A. Suraweera, H. Q. Ngo, T. Q. Duong, C. Yuen, and E. G. Larsson},
  ``{Multi-pair amplify-and-forward relaying with very large antenna arrays},''
  in \emph{Proc. IEEE Int. Conf. on Commun. (ICC)}, Budapest, Hungary, June
  2013, pp. 4635--4640.

\bibitem{Hoydis-13JSAC}
{J. Hoydis, S. ten Brink, and M. Debbah}, ``{Massive MIMO in the UL/DL of
  cellular networks: How many antennas do we need?}'' \emph{IEEE J. Sel. Areas
  Commun.}, vol.~31, no.~2, pp. 160--171, Feb. 2013.

\bibitem{ZhangJun-13JSAC}
{J. Zhang, C.-K. Wen, S. Jin, X. Q. Gao, and K.-K. Wong}, ``{On capacity of
  large-scale MIMO multiple access channels with distributed sets of correlated
  antennas},'' \emph{IEEE J. Sel. Areas Commun.}, vol.~31, no.~2, pp. 133--148,
  Feb. 2013.

\bibitem{Larsson-14ComM}
{E. G. Larsson, O. Edfors, F. Tufvesson, and T. L. Marzetta}, ``{Massive MIMO
  for next generation wireless systems},'' \emph{IEEE Commun. Mag.}, vol.~52,
  no.~2, pp. 186--195, Feb. 2014.

\bibitem{Lu-14JSTSP}
{L. Lu, G. Y. Li, A. L. Swindlehurst, A. Ashikhmin, and R. Zhang}, ``{An
  overview of massive MIMO: Benefits and challenges},'' \emph{IEEE J. Sel.
  Topics Sig. Proc.}, vol.~8, no.~5, pp. 742--758, Oct. 2014.

\bibitem{Sanguinetti-15JSAC}
{L. Sanguinetti, A. Moustakas, and M. Debbah}, ``{Interference Management in 5G
  Reverse TDD HetNets with Wireless Backhaul: A Large System Analysis},''
  \emph{IEEE J. Sel. Areas Commun.}, vol.~33, no.~6, pp. 1187--1200, June 2015.

\bibitem{Sadeghi-15GC}
{M. Sadeghi, C. Yuen, and Y. H. Chew}, ``{Multi-cell multi-group massive MIMO
  multicasting: An asymptotic analysis},'' in \emph{Proc. IEEE Global Commun.
  Conf. (GLOBECOM)}, San Diego, CA, Dec. 2015, pp. 1--6.

\bibitem{ZhangJun-15TWC}
{J. Zhang, C.-K. Wen, C. Yuen, S. Jin, and X. Q. Gao}, ``{Large system analysis
  of cognitive radio network via partially-projected regularized zero-forcing
  precoding},'' \emph{IEEE Trans. Wireless Commun.}, vol.~14, no.~9, pp.
  4934--4947, Sep. 2015.

\bibitem{Yang-13JSAC}
{H. Yang and T. L. Marzetta}, ``{Performance of conjugate and zero-forcing
  beamforming in large-scale antenna systems},'' \emph{IEEE J. Sel. Areas
  Commun.}, vol.~31, no.~2, pp. 172--179, Feb. 2013.

\bibitem{Gao-15CL}
{Z. Gao, L. Dai, C. Yuen, and Z. Wang}, ``{Asymptotic orthogonality analysis of
  time-domain sparse massive MIMO channels},'' \emph{IEEE Commun. Letters},
  vol.~19, no.~10, pp. 1826--1829, Oct. 2015.

\bibitem{Chen-15TWC}
{ X. Chen, L. Lei, H. Zhang, and C. Yuen}, ``{Large-scale MIMO relaying
  techniques for physical layer security: AF or DF?}'' \emph{IEEE Trans.
  Wireless Commun.}, Sep. 2015.

\bibitem{Fernandes-13JSAC}
{F. Fernandes, A. Ashikhmin, and T. L. Marzetta}, ``{Inter-Cell Interference in
  Noncooperative TDD Large Scale Antenna Systems},'' \emph{IEEE J. Sel. Areas
  Commun.}, vol.~31, no.~2, pp. 192--201, Feb. 2013.

\bibitem{Liu-14TSP}
{A. Liu and V. N. Lau}, ``{Joint power and antenna selection optimization in
  large cloud radio access networks},'' \emph{IEEE Trans. Sig. Proc.}, vol.~62,
  no.~5, pp. 1319--1328, Mar. 2014.

\bibitem{Joung-14JSTSP}
{J. Joung, Y. K. Chia, and S. Sun}, ``{Energy-efficient, large-scale
  distributed-antenna system (L-DAS) for multiple users},'' \emph{IEEE J. Sel.
  Topics Sig. Proc.}, vol.~8, no.~5, pp. 954--965, Oct. 2014.

\bibitem{LinYicheng-14TWC}
{Y. C. Lin and W. Yu}, ``{Downlink spectral efficiency of distributed antenna
  systems under a stochastic model},'' \emph{IEEE Trans. Wireless Commun.},
  vol.~13, no.~12, pp. 6891--6902, Dec. 2014.

\bibitem{Truong-13AC}
{K. T. Truong and R. W. Heath}, ``{The viability of distributed antennas for
  massive MIMO systems},'' in \emph{Proc. the Asilomar Conf. on Sig., Systems,
  and Computers}, Pacific Grove, CA, Nov. 2013, pp. 3--6.

\bibitem{Onireti-13TCOM}
{O. Onireti, F. Heliot, and M. A. Imran}, ``{On the energy efficiency-spectral
  efficiency trade-off of distributed MIMO systems},'' \emph{IEEE Trans.
  Commun.}, vol.~61, no.~9, pp. 3741--3753, Sep. 2013.

\bibitem{Wangdongming-13JSAC}
{D. M. Wang, J. Z. Wang, X. H. You, Y. Wang, M. Chen, and X. Y. Hou},
  ``{Spectral efficiency of distributed MIMO systems},'' \emph{IEEE J. Sel.
  Areas Commun.}, vol.~31, no.~10, pp. 2112--2127, Oct. 2013.

\bibitem{Sun-15IET}
{Q. Sun, S. Jin, J. Wang, Y. Zhang, X. Q. Gao, and K.-K. Wong}, ``{Downlink
  massive distributed antenna systems scheduling},'' \emph{IET Commun.},
  vol.~9, no.~7, pp. 1006--1016, May 2015.

\bibitem{Heath-11TSP}
{R. W. Heath, T. Wu, Y. H. Kwon, and A. Soong}, ``{Multiuser MIMO in
  distributed antenna systems with out-of-cell interference},'' \emph{IEEE
  Trans. Sig. Proc.}, vol.~59, no.~10, pp. 4885--4899, Oct. 2011.

\bibitem{Dailin-14TWC}
{L. Dai}, ``{An uplink capacity analysis of the distributed antenna system
  (DAS): From cellular DAS to DAS with virtual cells},'' \emph{IEEE Trans.
  Wireless Commun.}, vol.~13, no.~5, pp. 2717--2731, May 2014.

\bibitem{Ngo-13TCOM}
{H. Q. Ngo, E. G. Larsson, and T. L. Marzetta}, ``{Energy and spectral
  efficiency of very large multiuser MIMO systems},'' \emph{IEEE Trans.
  Commun.}, vol.~61, no.~4, pp. 1436--1449, Apr. 2013.

\bibitem{Liu-15arXiv}
\BIBentryALTinterwordspacing
{W. J. Liu, S. Q. Han, and C. Y. Yang}, ``{Is massive MIMO energy efficient?}''
  [Online]. Available: \url{http://arxiv.org/abs/1505.07187}
\BIBentrySTDinterwordspacing

\bibitem{EmilB-15TWC}
{E. Bj\"{o}rnson, L. Sanguinetti, J. Hoydis, and M. Debbah}, ``{Optimal design
  of energy-efficient multi-user MIMO systems: Is massive MIMO the answer?}''
  \emph{IEEE Trans. Wireless Commun.}, vol.~14, no.~6, pp. 3059--3075, June
  2015.

\bibitem{Yang-15VTC}
{H. Yang and T. L. Marzetta}, ``{Energy efficient design of massive MIMO: How
  many antennas?}'' in \emph{Proc. Vehicular Technology Conf. (VTC Spring)},
  Glasgow, UK, May. 2015, pp. 1--5.

\bibitem{Mohammed-14TC}
{S. K. Mohammed}, ``{Impact of transceiver power consumption on the energy
  efficiency of zero-forcing detector in massive MIMO systems},'' \emph{IEEE
  Trans. Commun.}, vol.~62, no.~11, pp. 3874--3890, Nov. 2014.

\bibitem{Hechunglong-12VTC}
{C. L. He, B. Sheng, P. C. Zhu, and X. H. You}, ``{Energy efficient comparison
  between distributed MIMO and co-located MIMO in the uplink cellular
  systems},'' in \emph{Proc. Vehicular Technology Conf. (VTC Fall)}, Quebec
  City, QC, Sep. 2012, pp. 1--5.

\bibitem{Zhao-13VTC}
{L. Zhao, H. Zhao, F. L. Hu, K. Zheng, and J. X. Zhang}, ``{Energy efficient
  power allocation algorithm for downlink massive MIMO with MRT precoding},''
  in \emph{Proc. Vehicular Technology Conf. (VTC Fall)}, Las Vegas, NV, Sep.
  2013, pp. 1--5.

\bibitem{Chien-16arXiv}
\BIBentryALTinterwordspacing
{T. V. Chien, E. Bj\"{o}rnson, and E. G. Larsson}, ``{Joint power allocation
  and user association optimization for massive MIMO systems}.'' [Online].
  Available: \url{http://arxiv.org/abs/1601.02436}
\BIBentrySTDinterwordspacing

\bibitem{Li-14WCNC}
{P. R. Li, T. S. Chang, and K. T. Feng}, ``{Energy-efficient power allocation
  for distributed large-scale MIMO cloud radio access networks},'' in
  \emph{Proc. Wireless Commun. and Networking Conf. (WCNC)}, Istanbul, Apr.
  2014, pp. 1856--1861.

\bibitem{Wagner-12IT}
{S. Wagner, R. Couillet, M. Debbah, and D. T. M. Slock}, ``{Large system
  analysis of linear precoding in correlated MISO broadcast channels under
  limited feedback},'' \emph{IEEE Trans. Info. Theory}, vol.~58, no.~7, pp.
  4509--4537, Jul. 2012.

\bibitem{ZhangJun-13TWC}
{J. Zhang, C.-K. Wen, S. Jin, X. Q. Gao, and K.-K. Wong}, ``{Large system
  analysis of cooperative multi-cell downlink transmission via regularized
  channel inversion with imperfect CSIT},'' \emph{IEEE Trans. Wireless
  Commun.}, vol.~12, no.~10, pp. 4801--4813, Oct. 2013.

\bibitem{Kay-93}
{S. M. Kay}, \emph{{Fundamentals of Statistical Signal Processing: Estimation
  Theory}}.\hskip 1em plus 0.5em minus 0.4em\relax Englewood Cliffs, NJ:
  Prentice Hall, 1993.

\bibitem{Marzetta-06ACSSC}
{T. L. Marzetta}, ``{How much training is required for multiuser MIMO?}'' in
  \emph{Proc. Asilomar Conf. Sig., Syst., Comput. (ACSSC)}, Pacific Grove, CA,
  Oct. 2006, pp. 359--363.

\bibitem{Jose-11TWC}
{J. Jose, A. Ashikhmin, T. L. Marzetta, and S. Vishwanath}, ``{Pilot
  contamination and precoding in multi-cell TDD systems},'' \emph{IEEE Trans.
  Wireless Commun.}, vol.~10, no.~8, pp. 2640--2651, Aug. 2011.

\bibitem{Ngo-11ICASSP}
{H. Q. Ngo, T. L. Marzetta, and E. G. Larsson}, ``{Analysis of the pilot
  contamination effect in very large multicell multiuser MIMO systems for
  physical channel models},'' in \emph{Proc. Acoustics, Speech and Sig. Proc.
  (ICASSP)}, Prague, Czech Republic, May 2011, pp. 3464--3467.

\bibitem{Hungerford-96}
{H. Hungerford}, \emph{{Abstract Algebra: An Introduction}}.\hskip 1em plus
  0.5em minus 0.4em\relax Brooks Cole, 1996.

\bibitem{Kumar-11WCVT}
{R. Kumar and J. Gurugubelli}, ``{How green the LTE technology can be?}'' in
  \emph{Proc. Int. Conf. on Wireless Commun., Veh. Techn., Inform. Theory and
  Aerosp. Electron. Syst. Techn.}, Chennai, Feb. 2014, pp. 1--5.

\bibitem{Desset-14GreenComm}
{C. Desset, B. Debaillie, and F. Louagie}, ``{Modeling the hardware power
  consumption of large scale antenna systems},'' in \emph{Invited at Proc. IEEE
  Online GreenComm}, Tucson, AZ, Nov. 2014, pp. 1--6.

\end{thebibliography}

\end{document}